\RequirePackage{snapshot}       
\documentclass[a4paper,UKenglish]{article}
\usepackage{microtype}

\usepackage{authblk}
\usepackage{amsmath,amsfonts,amsopn,amsbsy,amsthm}
\usepackage{hyperref}
\usepackage[margin=1.4in]{geometry}
\usepackage{graphicx}
\usepackage[utf8]{inputenc}
\newtheorem{theorem}{Theorem}
\newtheorem{lemma}[theorem]{Lemma}

\date{}

\title{Large Peg-Army Maneuvers}

\usepackage{array}
\usepackage{xspace}
\newtheorem{claim}{Claim}

\newcommand{\var}{0.47}

\newcommand{\hide}[1]{}
\newcommand{\pcsat}{\textsc{PCSat}\xspace}
\newcommand{\pegr}{\textsc{Solitaire-Reachability}\xspace}
\newcommand{\pegar}{\textsc{Solitaire-Army}\xspace}
\newcommand{\np}{\textup{NP}}
\newcommand{\srsap}{\textsc{R-Solitaire-Army}\xspace}
\newcommand{\vtrue}{\textsc{True}\xspace}
\newcommand{\vfalse}{\textsc{False}\xspace}

\def\x{{\bf x}}
\def\y{{\bf y}}
\def\z{{\bf z}}
\def\b{{\bf b}}
\def\c{{\bf c}}

\def\u{{\bf u}}
\def\v{{\bf v}}
\def\U{\mathcal{U}}
\def\uno{{\bf 1}}
\def\zero{{\bf 0}}

\author[1]{Luciano Gualà}
\author[2]{Stefano Leucci}
\author[2]{Emanuele Natale}
\author[1]{Roberto Tauraso}
\affil[1]{Università di Roma Tor Vergata. \texttt{\{guala, tauraso\}@mat.uniroma2.it}}
\affil[2]{Sapienza Università di Roma. \texttt{\{leucci, natale\}@di.uniroma1.it}}

\begin{document}

\maketitle

\begin{abstract}
Despite its long history, the classical game of peg solitaire continues to
attract the attention of the scientific community. In this paper, we consider
two problems with an algorithmic flavour which are related with this game,
namely Solitaire-Reachability and Solitaire-Army. In the first one, we show that deciding
whether there is a sequence of jumps which allows a given initial configuration
of pegs to reach a target position is NP-complete. Regarding Solitaire-Army, the aim
is to successfully deploy an army of pegs in a given region of the board in
order to reach a target position. By solving an auxiliary problem with relaxed
constraints, we are able to answer some open questions raised by
Cs\'ak\'any and Juh\'asz (Mathematics Magazine, 2000).
\end{abstract}

\section{Introduction}

\begin{quotation}
    \emph{Not so very long ago there became widespread an excellent kind of
    game, called Solitaire, where I play on my own, but as if with a friend
    as witness and referee to see that I play correctly. A board is filled
    with stones set in holes, which are to be removed in turn, but none
    (except the first, which may be chosen for removal at will) can be
    removed unless you are able to jump another stone across it into an
    adjacent empty place, when it is captured as in Draughts. He who
    removes all the stones right to the end according to this rule, wins;
    but he who is compelled to leave more than one stone still on the
    board, yields the palm. This game can more elegantly be played
    backwards, after one stone has been put at will on an empty board, by
    placing the rest with it, but the same rule being observed for the
    addition of stones as was stated just above for their removal. Thus we
    can either fill the board, or, what would be more clever, shape a
    predetermined figure from the stones; perhaps a triangle, a
    quadrilateral, an octagon, or some other, if this be possible; but such
    a task is by no means always possible: and this itself would be a
    valuable art, to foresee what can be achieved; and to have some way,
    particularly geometrical, of determining this.}
        -- {Gottfried Wilhelm Leibniz}\footnote{Original Latin text in
            \emph{Miscellanea Berolinensia} \textbf{1} (1710) 24, the given
            translation is from \cite{beasley_ins_1985}.}
\end{quotation}
\thispagestyle{empty}

\smallskip
In this work we investigate some computational and mathematical aspects of the
classical game known as \emph{peg solitaire}.
In peg solitaire, we have a grid graph (the \emph{board}) on each of whose nodes (the \emph{holes}) there may be at most one \emph{peg}. The initial configuration of pegs evolves by performing one of the following four moves (the
\emph{jumps}):
for each triple of horizontally or vertically adjacent nodes, if the first and
the second nodes are occupied by pegs and there is no peg on the third one,
then we can remove the two pegs and place a new one on the third node (see Figure
\ref{fig:a}). A \emph{puzzle} of peg solitaire is defined by an initial and a final
configuration, and consists of finding a sequence of moves that transforms the
initial configuration into the final one.

\begin{figure}
	\centering
	\includegraphics[scale=0.7, viewport=145 700 500 760]{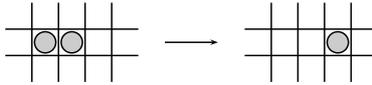}
	\caption{A peg solitaire move.}
	\label{fig:a}
\end{figure}

Because of the complexity generated by such simple rules
\cite{dijkstra_checkers_1992, gardner_unexpected_1991,
bialostocki_application_1998, bell_fresh_2007}, the game has attracted the
attention of many mathematically-inclined minds over its long history
(for which we refer the reader to Beasley's writings on the topic
\cite{bell_minimum_2006,beasley_solitaire:_2008, beasley_john_2015}).
In fact, we started with a
quotation from Leibniz, and the origin of the game may
well precede his time. Despite such a respectable age,
the problem of deciding which peg
solitaire puzzles can be solved is far from settled
\cite{beasley_solitaire:_2008}.

The present paper contributes to such
investigation by considering two specific problems:
\begin{itemize}
    \item \pegr: given a target position and an initial configuration of pegs
        on a finite board, we wish to determine whether there exists a sequence
        of moves that allows some peg to be placed in the given target
        position.
    \item \pegar: given a target position and a region of an infinite board,
        find an initial configuration of pegs inside that region, and a finite
        sequence of moves that allows some peg to be placed in the given target
        position.
\end{itemize}
We study the former problem from a computational point of view and we prove
that \pegr is \np-complete. This result, which we discuss in Section
\ref{sec:hardness}, is a significant step in understanding why peg solitaire
puzzles are intrinsically difficult. They are actually still a good testbed for
artificial intelligence techniques
\cite{kiyomi_integer_2000,jefferson_modelling_2006}. Indeed, the first
computational hardness result for this game was proved in 1990 in
\cite{uehara_generalized_1990}, but was limited to those peg solitaire puzzles
in which the final configuration is required to have only one peg (and hence
the goal was that of cleaning the entire board).

The \pegar problem has received a lot of attention from the mathematical and
game-connoisseur community \cite{bell_minimum_2006, aigner_moving_1997,
berlekamp_winning_2002}. In this body of works, the part of the board where no
pegs are allowed in the initial configuration is called \emph{desert}, and the
typical goal is reaching the farthest distance inside the desert.

In the classical example, introduced by J. H. Conway, the desert is a
half-plane. Conway devised an elegant potential function argument which shows
that, for any initial configuration, and no matter what sequence of moves one
may attempt, no peg can reach a distance larger than four in the
desert \cite{honsberger1976problem}.

Other shapes for the desert have been considered. For instance, in
\cite{csakany_solitaire_2000}, the authors focus on square-shaped and
rhombus-shaped deserts. Here the natural target position is the center of the
square/rhombus, and the goal is that of finding the largest size of the desert
for which the puzzle is solvable. Among other results, in
\cite{csakany_solitaire_2000} it is shown that the $9\times 9$ square-shaped
and $13\times 13$ rhombus-shaped deserts are solvable, while the $13 \times 13$
square-shaped and $17\times 17$ rhombus-shaped are not. Therefore the problems
whether the $11\times 11$ square-shaped and $15\times 15$ rhombus-shaped
deserts are solvable were left open\footnote{Note that the square/rhombus
should have the side of odd length in order for the center of the desert to be
well-defined.}.

In Section \ref{sec:army} we discuss our contribution to the \pegar problem. We
develop a general approach that can be used to attack such kind of puzzles, and
as a byproduct, we are able to show that both the $11\times 11$ square-shaped
and $15\times 15$ rhombus-shaped deserts are actually solvable. The main idea
underlying our technique is  considering an auxiliary problem where the
constraint that any node can have at most one peg is \emph{relaxed} to allowing
pegs to be stacked. We actually allow each node to have any integer number of
pegs, including negative. This setting has the advantage that the order of the
moves is immaterial. The auxiliary problem admits a natural compact Integer
Linear Programming (ILP) formulation which, when the puzzle is simple enough,
can be safely solved by means of an ILP solver. Once the relaxed problem is
solved we provide also an efficient algorithm which {\sl converts} the relaxed
solution into a solution for the original problem.
To appreciate the combinatorial beauty of our solutions,
we recommend to visit the gallery of animations provided at \url{http://solitairearmy.isnphard.com}.

Interestingly enough, an analogous relaxation was considered in
\cite{chung_pebbling_1995} for a similar problem where a different set of moves
is used (the so-called \emph{pebbling} game). The authors claim the equivalence
between the relaxed and the original problem (see Lemma 3 in
\cite{chung_pebbling_1995}) but the proof is omitted as an `easy induction
argument'. Apparently, E. W. Dijkstra disagrees with such statement: ``\emph{it
    would have been nice if `the easy induction argument' had been shown: a few
colleagues and I spent $1\frac 12$ hours on not finding it}''
\cite{dijkstra_only_1995}. Since our equivalence result holds for a large class
of moves which includes both pegs and pebbles, we also obtain their result as a
corollary.

\textbf{Other related results.} Our hardness reduction contributes to the line
of research investigating the computational complexity of combinatorial games
\cite{hearn_games_2009,
aloupis_classic_2015, guala_bejeweled_2014}. As for positive results regarding
peg solitaire, it has been shown that, on a rectangular board of fixed height,
the set of initial configurations that can be reduced to a single peg form a
regular language \cite{goos_peg-solitaire_1997, moore_one-dimensional_2000}.
Finally, the idea of formulating peg solitaire puzzles as integer linear
programs
was already suggested in \cite{beasley_ins_1985} and effectively used in
\cite{kiyomi_integer_2000}.

\section{Our Hardness Reduction}
\label{sec:hardness}

\subsection{Overview of the Reduction}
Here we consider the problem \pegr: given an initial configuration of pegs on a
finite board, it asks to decide whether there exists a sequence of moves that
cause a peg to be placed in a given target position. We prove that \pegr is
\np-complete.

Our reduction is from the \emph{planar circuit satisfiability problem} (\pcsat
for short): in \pcsat we are given a \emph{boolean network} represented as a
planar directed acyclic graph $G$ having a single \emph{sink} vertex $t \in
V(G)$ (i.e., a vertex having out-degree $0$).
Each other vertex in $V(G)$ is either an \emph{input} vertex or a \emph{NAND} vertex.
The vertex $t$ is required to have in-degree $1$, input vertices must be
sources in $G$ (i.e. their in-degree must be $0$), while NAND vertices must
have in-degree exactly $2$ and out-degree at least $1$.
The problem consists in determining whether it is possible to assign a truth
value $\pi(u) \in \{\vtrue, \vfalse \}$ to each vertex $u \in G$ in order to
satisfy the following properties:
\begin{itemize}
    \item the truth value assigned to a NAND vertex $u$ is the NAND of the
        truth values of its two in-neighbors, i.e., if $v_1$ and $v_2$ are the
        in-neighbors of $u$, we have $\pi(u) = \neg ( \pi(v_1) \wedge
        \pi(v_2))$;
    \item the value assigned to the sink vertex $t$ coincides with the truth
        value of its only in-neighbor;
	\item $\pi(t) = \vtrue$.
\end{itemize}

Notice that the assignment $\pi(\cdot)$ is completely determined by the truth
values of the input vertices. It is well known that this problem is
\np-hard.\footnote{See, for example, the notes of lecture 6 of the course ``Algorithmic Lower Bounds: Fun with Hardness Proofs'' by Prof. Erik Demaine (\url{http://courses.csail.mit.edu/6.890/fall14/lectures/}). } Clearly, the graph can be thought as a \emph{boolean circuit}
consisting of \emph{links} and \emph{gates} which computes a boolean output as
a function of the circuit's inputs $x_1, x_2, \dots$.

\begin{figure}
	\centering
	\includegraphics[scale=0.8]{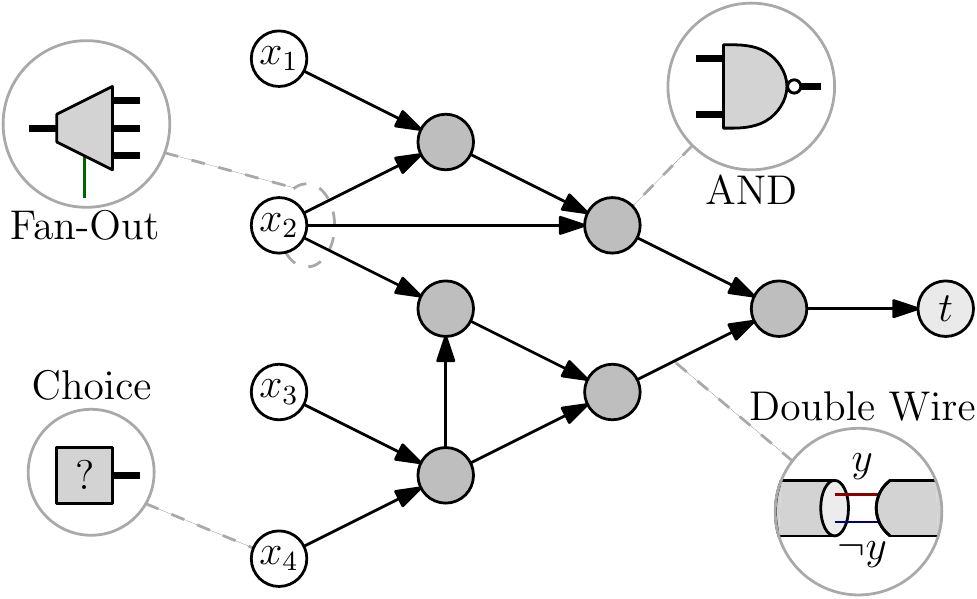}
	\caption{An instance of \pcsat corresponding to the formula $( \neg x_1 \vee \neg x_2) \wedge (x_2 \vee \neg x_3 \vee \neg x4)$. Gadgets that we need to implement are highlighted.}
	\label{fig:overview}
\end{figure}

Given any instance of \pcsat, we will build an instance of \pegr which
simulates the behavior of such a circuit. We encode the circuit's mechanics
using \emph{dual-rail logic}: each edge of the network will be transformed
into a \emph{dual-rail wire} consisting of two \emph{single wires} which will
always ``carry'' opposite boolean values. In order to perform such a
transformation we will make use of some gadgets:  the \emph{choice gadget} will
be used to encode input vertices, the \emph{NAND gadget} will represent a NAND
vertex, and the \emph{fan-out gadget} will allow to split a single double-wire
into multiple double-wires to be fed as input into other gadgets.
For technical reasons, the fan-out gadget will output an additional boolean
\emph{control signal} on a dedicated single wire, which we will call
\emph{control wire}. Intuitively, if we require this control signal to be
\vtrue, then the correct operation of the fan-out gadget is guaranteed.
All these control wires can be safely brought outside of the area of the board
that contains the gadgets and all the other wires. Once this has been done, we
can finally AND together the boolean signals carried by all these
control lines along with the signal carried by the double wire corresponding to
the unique edge entering the output vertex $t$.
Since $G$ is planar we will have no intersections between double-wires.
However, the same does not hold for control wires. Whenever a control wire
intersects a double wire we will use a suitable additional gadget which we
call \emph{control-crossover} that ensure that, whenever the control signal is
\vtrue, the signal carried by the double wire will not be affected.

Figure~\ref{fig:overview} shows a possible instance of \pcsat and highlights the gadgets we need to implement, while Figure~\ref{fig:overview2} is a high level picture of the associated instance of \pegr.
Clearly, the circuit shown in Figure~\ref{fig:overview2} will be implemented as a certain configuration of pegs on a board. The output of the last AND gate will correspond to the target position of \pegr: a peg can be brought to the target position iff there is a truth assignment of the inputs that causes the circuit to output \vtrue.

\begin{figure}
	\centering
	\includegraphics[width=0.75\textwidth]{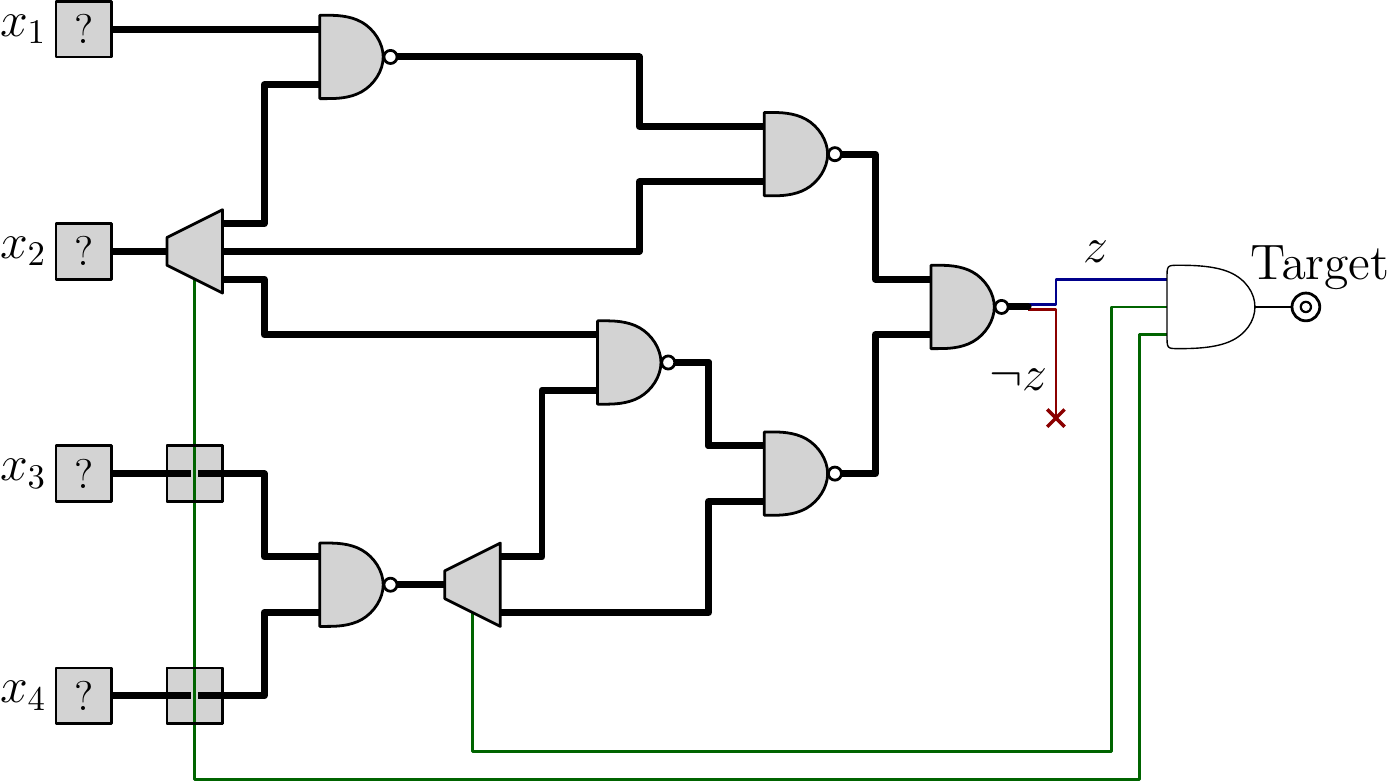}
	\caption{A high-level picture of the instance of \pegr corresponding to the instance of \pcsat shown in  Figure~\ref{fig:overview}. Double wires are bold while single wires are thin. Control wires (in green) are ANDed together with the single wire carrying the positive signal $z$ of the circuit output. The negated signal $\neg z$ is ignored.}
	\label{fig:overview2}
\end{figure}
	
\subsection{Description of the Gadgets}		

In this section we describe the gadgets used by our reduction. We will need to design gates for both binary and dual-rail logic. We will use the following color convention: the former will be colored white, while the latter will be in gray.

\subsubsection{Binary Logic Gadgets}

We start by describing the gadgets of our construction that deal with binary
logic. We have single wires and binary logic gates. Intuitively a single wire
``carries'' a boolean signal from its first endpoint to the other. The truth
value of the signal is encoded as follows: if there is a peg on the first
endpoint then the signal is \vtrue, and the wire will allow the peg to be moved
to the other endpoint via a sequence of moves. On the converse, if there is no
peg on the first endpoint then no move can be made and therefore no peg can be
placed on the second endpoint, which encodes a \vfalse signal. The
implementation of a single wire is straightforward and it is shown in
Figure~\ref{fig:bl-wire-half-crossover}. In order to deal with parity issues we
can use the Shift gadget (also shown in
Figure~\ref{fig:bl-wire-half-crossover}) that allows to shift a wire by one
row/column of the board. This gadget can also be used to force the signal to
flow in one direction.

As far as gates are concerned, each gate takes one or more boolean inputs
(which are again encoded by the presence or absence of pegs in certain input
positions), and has one or more boolean outputs which are a function of the
inputs. If an output is \vtrue, this means that there exists a sequence of
moves which places a peg on the output position of the gate. Otherwise, if an
output is \vfalse, no sequence of moves can place a peg on the output position
of the gate.
	
\paragraph{Binary Logic Half-Crossover}

Intuitively, this gadget allows for two single wires to safely cross each other if at least one of them is \vfalse. When both signals are \vtrue they cannot both cross, and one of them will become \vfalse.

More formally, the half-crossover has two inputs $x$ and $y$ and two outputs $x'$ and $y'$. If $x$ and $y$ are not both true, then $x'$ is true iff $x$ is true, and $y'$ is true iff $y$ is true. Otherwise, if both $x$ and $y$ are true, then only one of $x'$ and $y'$ can be set to true.
The gadget is shown in Figure~\ref{fig:bl-wire-half-crossover}.

\begin{figure}
	\centering
	\includegraphics[width=0.9\textwidth]{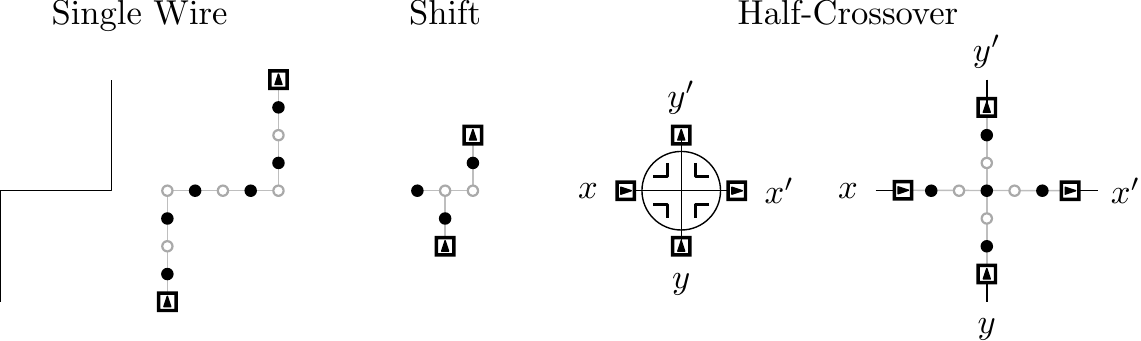}
	\caption{Symbols and implementations of the wire (left), shift (center), and half-crossover (right). In the half-crossover, if the inputs $x$ and $y$ are not both true, then $x'=x$ and $y'=y$. Otherwise either $x'$ or $y'$ will be true (but not both).}
	\label{fig:bl-wire-half-crossover}
\end{figure}

\paragraph{Binary Logic AND and Binary Logic OR}

The binary logic AND and OR gadgets have two inputs and one output and their implementations are shown in Figure~\ref{fig:bl-and-or}. In the AND gadget it is possible to place a peg on the output position only when both pegs are present on the input positions while, in the OR gadget, it is possible to output a peg iff at least one input peg is present. Clearly, multiple copies of both the AND and the OR gadgets can be chained together in order to simulate AND and OR gates with multiple inputs.

\begin{figure}
	\centering
	\includegraphics[scale=1]{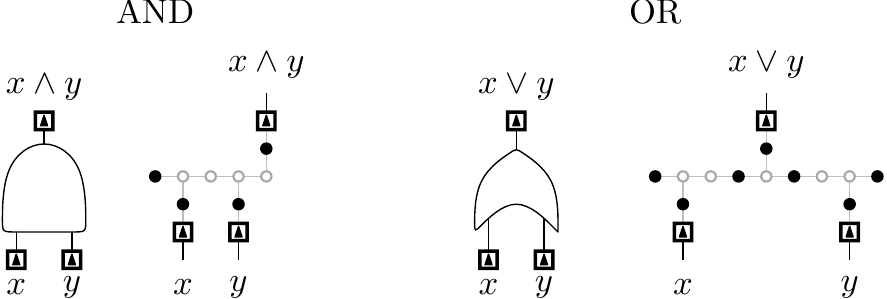}
	\caption{Implementation of the AND and OR binary logic gates.}
	\label{fig:bl-and-or}
\end{figure}

\paragraph{Binary Logic Fan-Out}

This gadget allows to duplicate an input signal $x$. For technical reasons it also has an additional output signal $c$ that we call control signal. Whenever $c$ is $\vtrue$, this gate acts as a classical fan-out: the values of the two outputs $x'$ and $x''$ (see Figure~\ref{fig:bl-AOA-fan-out}) will coincide to the value of $x$. We will always require all the control lines to be \vtrue.
In order to implement this gate we need an additional gadget that we call AXB gate.

The AXB gate takes three inputs $a$, $x$, $b$ (in order) and computes a single output whose value is \vtrue iff $x$ is \vtrue or both $a$ and $b$ are \vtrue. The implementation is given in Figure~\ref{fig:bl-AOA-fan-out}. Notice that if $x$ is \vtrue, then a \vtrue signal can traverse the half-crossover and trigger the OR-gate, which outputs \vtrue. If $x$ is \vfalse, then the only way for the gate to output \vtrue is to have both $a$ and $b$ set to \vtrue. Indeed, in this case, $a$ can traverse the half-crossover and reach the AND gate together with $b$. Both this signals are need to cause the AND and the OR gates to output \vtrue.

Now we argue on the correctness of the fan-out gadget. Assume that the output control signal $c$ is $\vtrue$. This means that either $x$ is \vtrue and hence two pegs can be placed on the output positions of $x'$ and $x''$, or $x$ is \vfalse which means that the other two inputs of the AXB gate must both be true, which implies that no peg can reach $x'$ or $x''$.

\begin{figure}
	\centering
	\includegraphics[width=0.95\textwidth]{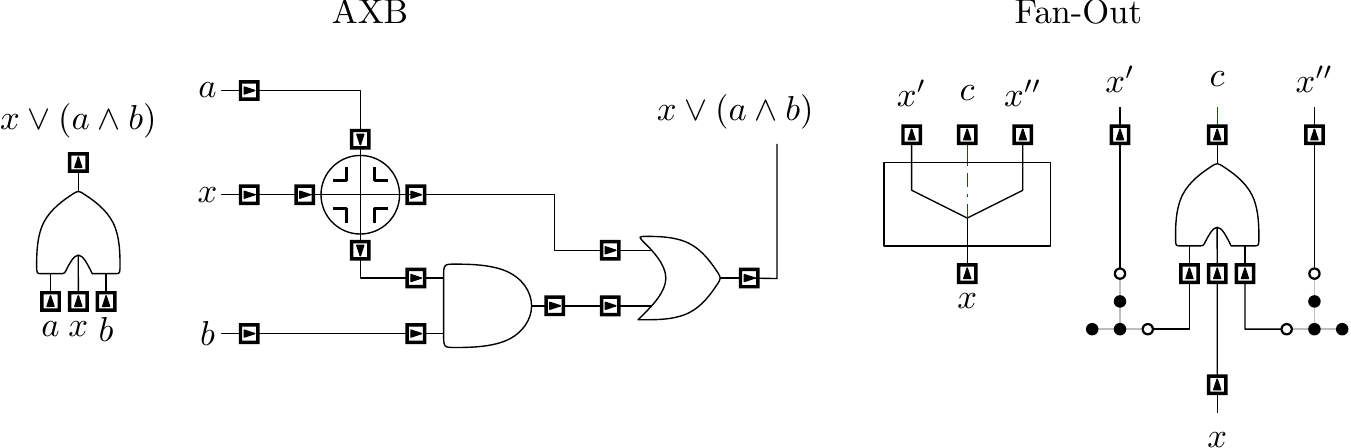}
	\caption{Implementation of the AXB and fan-out binary logic gates.}
	\label{fig:bl-AOA-fan-out}
\end{figure}

\paragraph{Binary Logic Control-Crossover}

This gadget allows a control signal $c$ to safely cross a the signal $x$ carried by single wire. The gadget has two outputs $c'$ and $x'$ and its implementation is shown in Figure~\ref{fig:bl-control-crossover}. If  $c'$ is required to be \vtrue, then $c$ must be also true and $x'$ will be \vtrue iff $x$ is \vtrue. Notice that if $c'=\vfalse$ then $x'$ can be freely set to $\vtrue$ or $\vfalse$. However, we will always require $c'$ to be set to \vtrue.

We now argue on the correctness of the gadget. We only need to consider the case $c'=\vtrue$. It is easy to see that, in order for $c'$ to be \vtrue, $c$ must be \vtrue as well. Indeed if $c=\vfalse$, then at least one of the two inputs of the AND gate must be false. Let us assume $c=\vtrue$. Either $x=\vtrue$ or $x=\vfalse$. In the first case, the gate can output $x'=\vtrue$ and $c'=\vtrue$: we can use the peg from $c$ to allow the peg from $x$ to reach the fan-out gate which can now output three \vtrue signals. Notice that this requires the peg from $x$ to jump over the peg from $c$.
In the other case, i.e. $x=\vfalse$, we have that the gate can output $x'=\vfalse$ and $c'=\vtrue$ by using the peg from $c$ to activate the OR gate, and the control signal from the fan-out gate to activate the end gate. Now, we only need to show that the gate cannot output $x'=\vtrue$ and $c'=\vtrue$. Indeed, as $c'=\vtrue$, the control signal from the fan-out gate must be true as well. The claim follows since, in this case, there is no way for a $\vtrue$ signal to reach the input of the fan-out gate.

\begin{figure}
	\centering
	\includegraphics[scale=1]{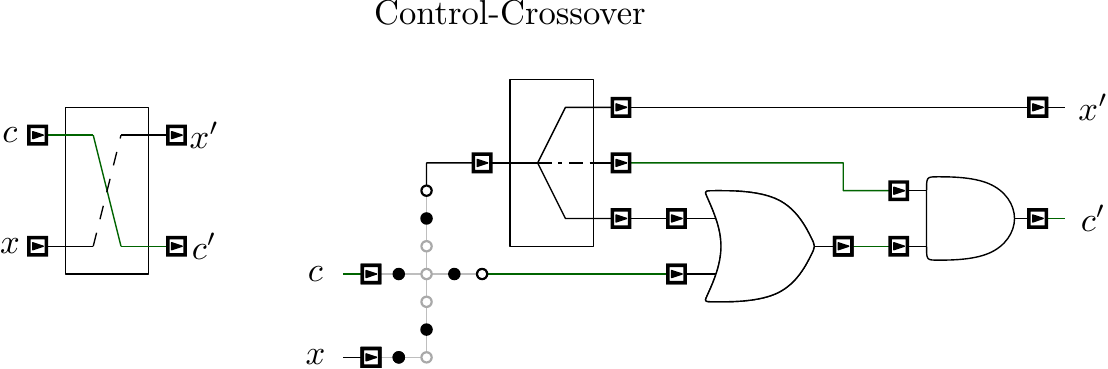}
	\caption{Implementation of the control-crossover gadget. If $c'=\vtrue$, then $c=\vtrue$ and $x'=x$. If $c'=\vtrue$ then $x'$ is unrelated to $x$ and can be either $\vtrue$ or $\vfalse$.}
	\label{fig:bl-control-crossover}
\end{figure}

\subsubsection{Dual-Rail Logic Gadgets}
	
\paragraph{Double Wire and Dual-Rail Logic Choice}

These two gadgets are straightforward and their implementation are shown in Figure~\ref{fig:drl-wire-choice}. Double wires will be drawn in bold to distinguish them from single wires. Moreover we will refer to the pair of signals $x$ and $\neg x$ as $\hat x$.
In the choice gadget, the central peg must jump over either the peg on its top (that corresponds to $x$) or the peg on its bottom (that corresponds to $\neg x$).

\begin{figure}
	\centering
	\includegraphics[scale=1]{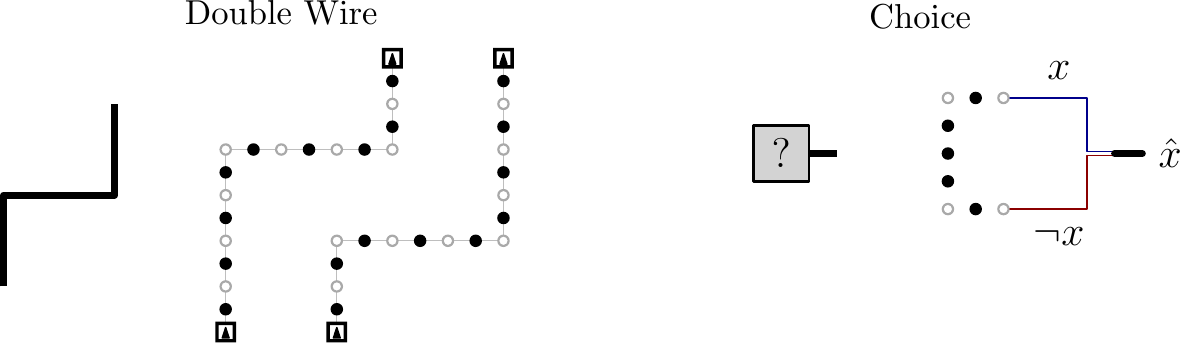}
	\caption{Implementation of a double wire and of the dual-rail logic choice gadget.}
	\label{fig:drl-wire-choice}
\end{figure}

\paragraph{Dual-Rail Logic NAND}

We implement this gate by separately building a dual-rail NOT gate and a a dual-rail AND gate, see Figure~\ref{fig:drl-and}.
To compute $\neg \hat x$ from $\hat x$, we just need to exchange the roles of the single wires corresponding to $x$ and $\neg x$, hence the NOT gate is actually just a half-crossover.
The AND gate computes the logic AND of $\hat x$ and $\hat y$ by separately computing $x \wedge y$ and $\neg (x \wedge y) = \neg x \vee \neg y$. Notice that, every time we use a half-crossover, at least one of its inputs is \vfalse.

\begin{figure}
	\centering
	\includegraphics[scale=1]{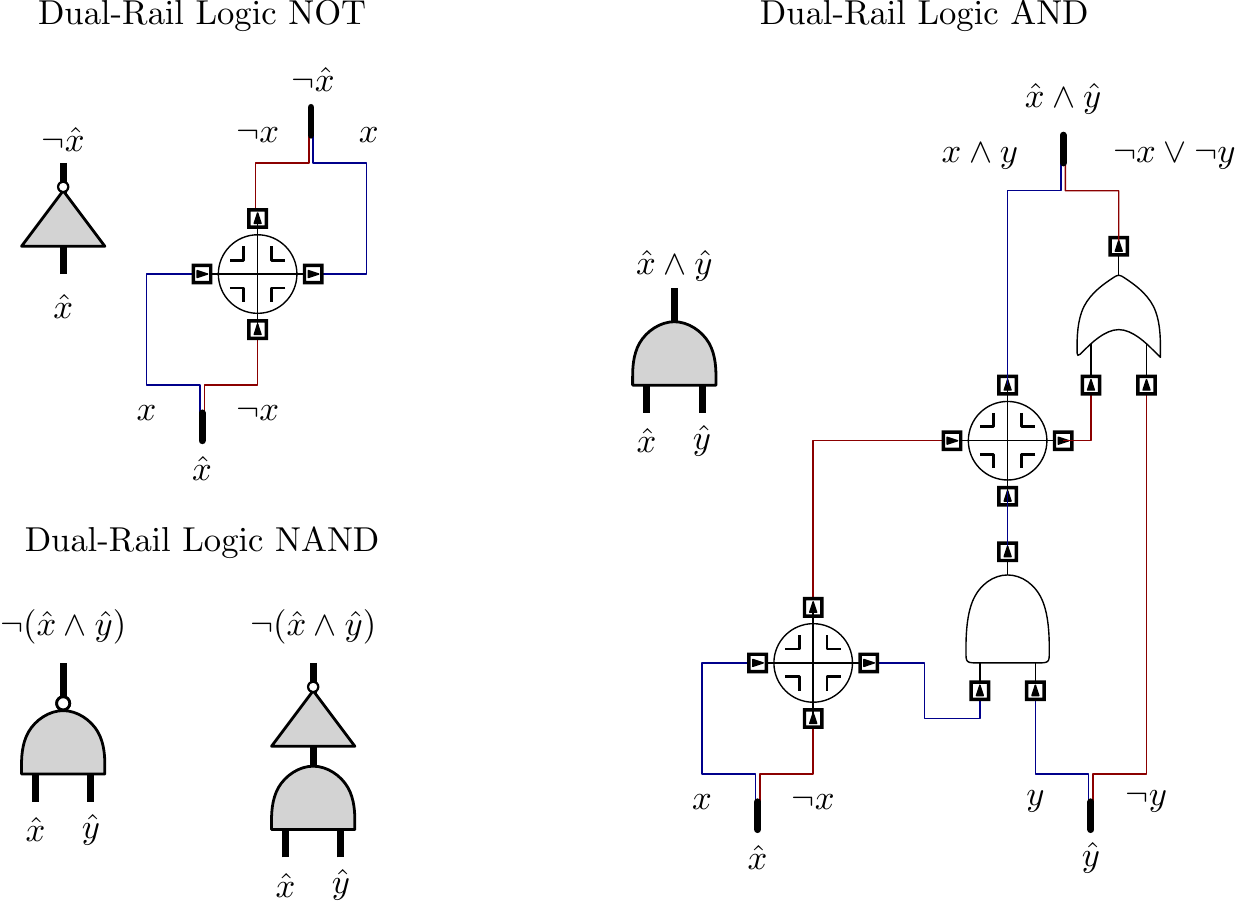}
	\caption{Implementation of the dual-rail logic NOT, AND, and NAND gates.}
	\label{fig:drl-and}
\end{figure}

\paragraph{Dual-Rail Logic Fan-Out}

This gate is shown in Figure~\ref{fig:drl-fan-out} and it is similar to the binary logic fan-out except that it works on double wires. The gadget makes use of two binary logic fan-outs whose control lines are crossed with the single wires using binary logic control crossovers. The additional control output $c$ will be \vtrue iff all the inner control lines are \vtrue as well, hence ensuring the correct operation of the gadget. Notice that we also need to cross a single wire carrying the signal $\neg x$ with another single wire carrying $\neg x$ but this can be safely done by using an half-crossover.

\begin{figure}
	\centering
	\includegraphics[scale=1]{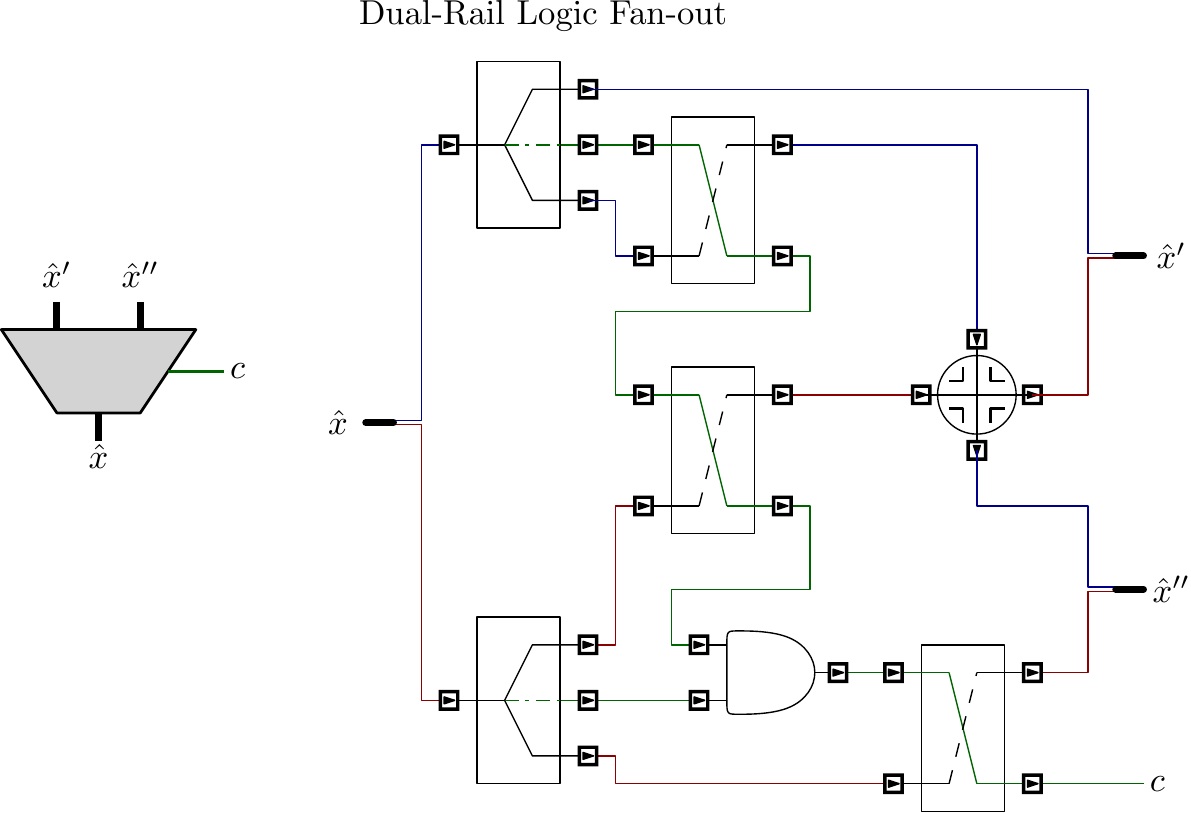}
	\caption{Implementation of the dual-rail logic fan-out. Whenever $c'=\vtrue$, the gadget duplicates $\hat x$ into $\hat x'$ and $\hat x''$.}
	\label{fig:drl-fan-out}
\end{figure}

\paragraph{Dual-Rail Logic Control-Crossover}

This gadget allows for a control wire to safely cross a double wire: if the control output $c'$ is \vtrue, then the control input $c$ is true as well and the output $x'$ is equal to the input $x$ (see Figure~\ref{fig:drl-control-crossover}). Its implementation is straightforward as it suffices to use two binary logic control-crossovers to cross $c$ with the single wires corresponding to $x$ and $\neg x$, respectively.

\begin{figure}
	\centering
	\includegraphics[scale=1]{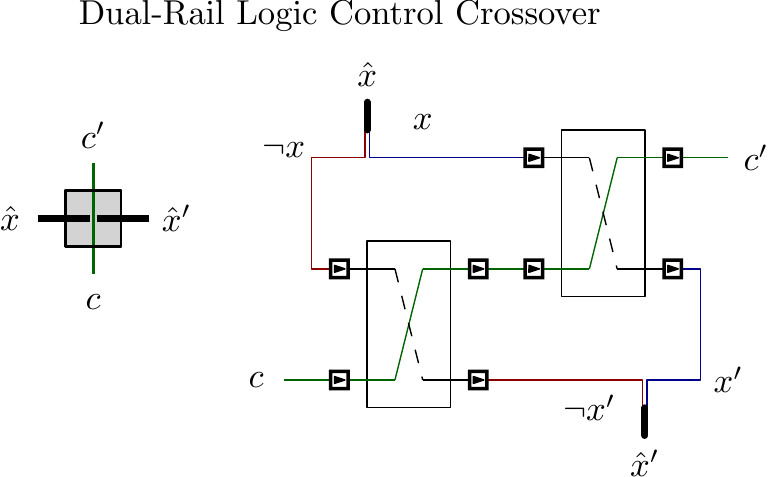}
	\caption{Implementation of the dual-rail logic control-crossover.}
	\label{fig:drl-control-crossover}
\end{figure}

\section{The Solitaire-Army problem}
\label{sec:army}

\subsection{Overview}

In this section we focus on the \pegar problem which, given a target position
$g$ and a region $R$ of an infinite board, requires finding an initial
configuration of pegs inside $R$, and a finite sequence of moves such that a peg is
finally placed in $g$. We will call \emph{desert} the part of the board where
no pegs are allowed in the initial configuration. Moreover, when we consider a
configuration of the board, for any position $p$, we denote by $n(p)$ the
number of pegs placed at $p$ in the configuration\footnote{We remark that in
\pegar it always holds  $n(p)\in \{0,1\}$.}.
\begin{table}
    \centering
    \begin{tabular}{|>{\centering\arraybackslash}m{6.6cm}|>{\centering\arraybackslash}m{2.8cm}|>{\centering\arraybackslash}m{2cm}|}
        \hline
        \textbf{Desert} & \textbf{Goal vertex} & \textbf{Number of moves} \\
        \hline
        Square $7 \times 7$ & Center of the desert & 15\\
        \hline
        Square $9 \times 9$ & Center of the desert & 39\\
        \hline
        Square $11 \times 11$ & Center of the desert & 212\\
        \hline
        Square $12 \times 12$ &
        \begin{minipage}[t]{0.2\textwidth}
            {One of the four centers of the desert\smallskip}
        \end{minipage} & 301\\
        \hline
        \begin{minipage}[t]{\var\textwidth}
            \flushleft {Square $11 \times 11$ at the border of the board,
              when the board is a half-plane \smallskip}
        \end{minipage} &
        Center of the desert & 246\\
        \hline
        \begin{minipage}[t]{\var\textwidth}
            \flushleft Square $11 \times 11$, when the board is the union of three
            half-planes tangent to the desert \smallskip
        \end{minipage}
            & Center of the desert & 241\\
        \hline
        \begin{minipage}[t]{\var\textwidth}
            \flushleft Rhombus $15\times 15$ (i.e. with axes of length $15$)
        \end{minipage}
            & Center of the desert & 176\\
        \hline
        \begin{minipage}[t]{\var\textwidth}
            \flushleft Rhombus $15\times 15$ at the border of the board,
            when the board is a diagonal half-plane\smallskip
        \end{minipage}
            & Center of the desert & 202\\
        \hline
        \begin{minipage}[t]{\var\textwidth}
            \flushleft Rhombus $15\times 15$,
            when the board is the union of three half-planes tangent to the desert \smallskip
        \end{minipage}
            & Center of the desert & 183\\
        \hline
    \end{tabular}
    \caption{Solutions to \pegar of particular interest.}
    \label{tab:sap}
\end{table}
When trying to solve an instance of \pegar, it is convenient to consider the
\emph{reversed} version of the game, as if we were rewinding a video of someone
trying to solve the puzzle. More formally, in the reverse \pegar, a move is
defined as follows. Let $\textbf{p} = (p_1, p_2, p_3 )$ be a generic triple of
vertically or horizontally adjacent positions $p_1$, $p_2$ and $p_3$. A move
consists of decreasing by one $n(p_3)$ and increasing by one both $n(p_1)$ and
$n(p_2)$, provided that $n(p_1)=n(p_2)=0$ and $n(p_3)=1$. Clearly, an instance
of \pegar is solvable if and only if, in the reversed version of the game, it
is possible to reach a configuration with no peg in the desert, starting from
an initial configuration with a single peg placed in $g$. In the rest of this
section, we will consider always the reversed version of \pegar and we will
refer to it as simply \pegar.

\begin{figure}
	\centering
	\includegraphics[scale=0.7]{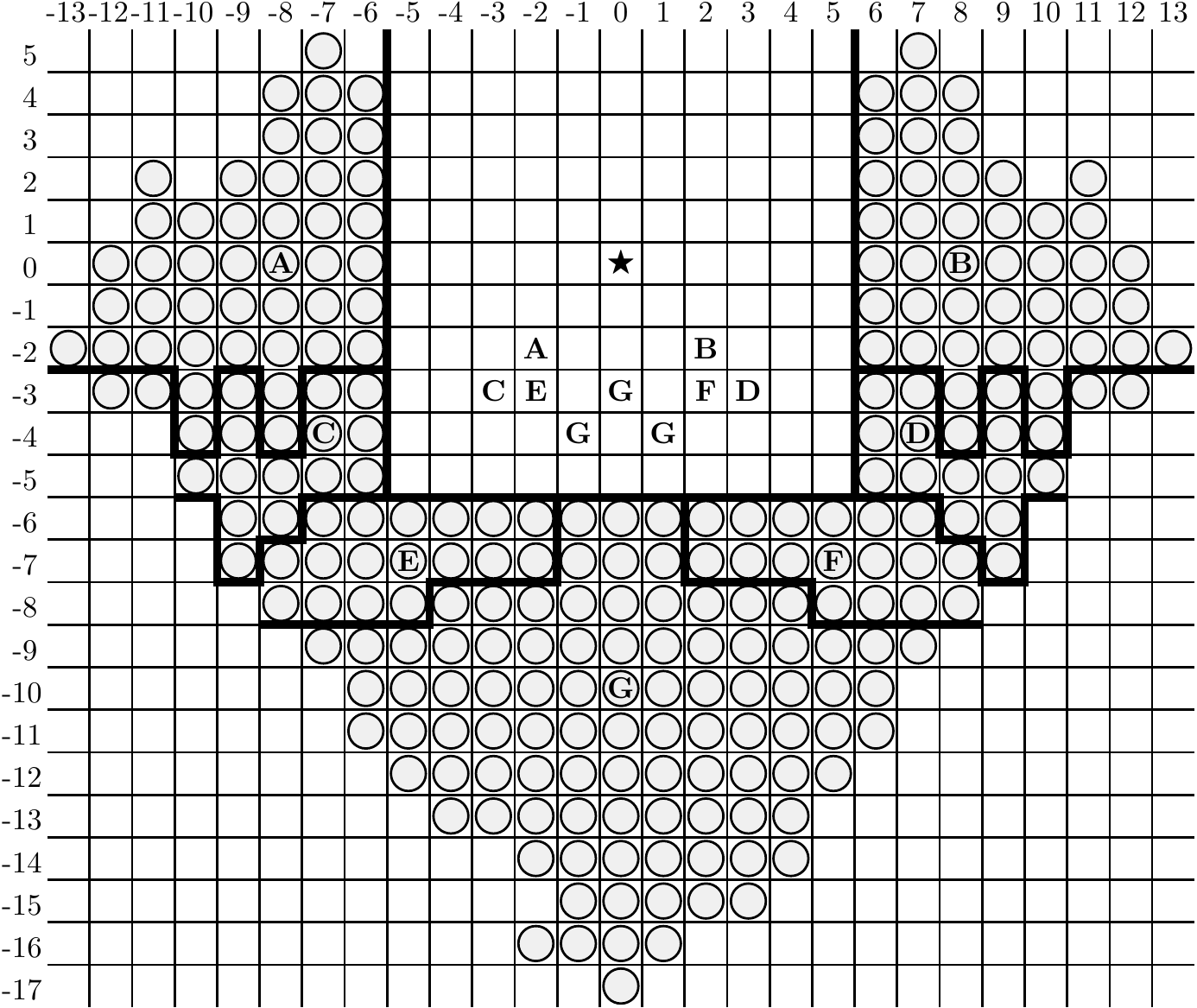}
    \caption{The initial deployment of the \pegar  which sends a peg to the
center of a square $11\times 11$ desert (this problem was left open in \cite{csakany_solitaire_2000}).
The army is divided into seven platoons each denoted by a letter.
The letters inside the square represent the positions reached by the
corresponding platoons. When those positions are taken, it is easy to
conquer the center denoted by a star.
The complete sequence of moves is given in the Section \ref{apx:moves}.}
\label{fig:peg}
\end{figure}

We are now ready to describe a general approach that can be conveniently used
to attack an instance of \pegar. One of the most annoying issue when trying to
solve the problem is the fact that the order of the moves does matter. To
circumvent this difficulty, we simplify the game by defining a \emph{relaxed}
version of \pegar that we call \srsap. In this variant, given a triple
$\textbf{p} = (p_1, p_2, p_3 )$ on the board, a move consists of decreasing
$n(p_3)$ by one and increasing both $n(p_1)$ and $n(p_2)$ by one, with no
constraints on the values of $n({p_1})$, $n({p_2})$ and $n({p_3})$. Therefore
we allow each position to have any integer number of pegs, including negative.
Here the goal is to reach a configuration in which there are zero pegs in the
desert and every position in the region $R$ is occupied by one or zero pegs.

Our main contribution of this section is the following theorem, which shows the
equivalence among the two versions of the game. The
proof is deferred to the next Section.

\begin{theorem}
    \srsap is solvable if and only if \pegar is solvable. Moreover, any solution for \srsap can be transformed in polynomial time into a solution for \pegar (with at most the same number of moves).
    \label{thm:SRT_iff_OP}
\end{theorem}

We emphasize that Theorem \ref{thm:SRT_iff_OP} actually holds for a wider class of games
(including the \emph{pebbling} game \cite{chung_pebbling_1995}).
More precisely, we can define a generalized (reversed) \pegar problem in which the set of moves is specified by a collection of tuples $\textbf{p}^{(i)}=(p^{(i)}_{1}, \dots, p^{(i)}_{\ell_i} )$
of vertices of the board. In this game a move consists
of removing a peg from the vertex $p^{(i)}_{\ell_i}$
and of adding a peg in the remaining vertices of the tuple, provided that $n({p^{(i)}_{\ell_i}}) = 1$ and
$n({p^{(i)}_{1}}) = \dots = n({p^{(i)}_{\ell_i-1}}) = 0$.

Since it turns out that
\srsap admits a compact integer linear programming formulation (ILP), then several \pegar puzzles can be solved by using a good ILP solver\footnote{We used Gurobi Optimizer, that Gurobi Optimization Inc. has gently made freely available to academic users for research purposes.}. Table \ref{tab:sap} summarizes some of the results that we obtained by using our general approach. We also provide a gallery of \emph{animated} solutions at \url{http://solitairearmy.isnphard.com}.
As an example, in Figure
\ref{fig:peg}, we illustrate a final configuration for the \pegar in which the initial target position is the center of a $11\times 11$ square-shaped desert.

\subsection{Proof of Theorem \ref{thm:SRT_iff_OP}}
\label{sec:proof_of_thm}

Recall that we are considering the reversed version of \pegar.
For any finite board, let $n$ be the number of positions.
Let $\zero$ be the $n\times 1$ vector whose entries are
all $0$s, and $\uno$ be the $n\times 1$ vector whose entries are
all $1$s. Let $\b$ be a $n\times 1$ integer vector with
$\zero\leq \b\leq \uno$ which represents the initial configuration of pegs.
Let $\c$ be a $n\times 1$  integer vector with $\zero\leq \c\leq \uno$ which
represents the constraints on the final configuration. In particular the zero entries of $\c$
specify the region that at the end should be cleared of pegs (the desert).

An $n\times m$ matrix $A=(a_{ij})$ defines the set of moves on the board: for
$j=1,\dots,m$, the $j$th move adds $a_{ij}\in\{0,1,-1\}$ pegs to the $i$th
position. The \emph{fundamental assumption} on this matrix is that each column
of $A$, a move, has at most a $-1$ entry. No further assumption on the
structure of matrix $A$ is made.

According of this notation, we state the relaxed problem (\srsap):
find a $m\times 1$ integer
vector $\x\geq \zero$ such that
\[\zero\leq A\x+\b\leq \c.\]
Note that the component $(A\x+\b)_i$ gives the number of pegs at the $i$th position after we applied
$(\x)_j$ times the $j$th move for $j=1,\dots,m$ to the initial configuration
$\b$.

On the other hand, the original problem \pegar  can be rephrased as follows. Find $m\times 1$ unit vectors
$\u_1,\u_2,\dots, \u_N$  (a {\sl unit vector} is a vector which has an entry $1$ and $0$s elsewhere) such that
\[
    \zero\leq A\x_k+\b\leq \uno\quad \mbox{for } 1\leq k\leq N-1 \quad
        \mbox{and}\quad \zero\leq A\x_N+\b\leq \c.
\]
where $\x_k:=\sum_{t=1}^k \u_t$. Hence if a solution $\x$ of \srsap is available, we need to
find a decomposition of $\x$ into an ordered sum of unit vectors which represents a specific sequence of feasible  moves for \pegar.

\smallskip

For a generic vector $\x$, we denote by $|\x|$ the sum of the entries of $\x$.
The following preliminary lemma describes a crucial property we will need later.

\begin{lemma} Let $\x$ be a solution of \srsap and let $\y$ be a $m\times 1$ integer vector such that
\[
    \zero\leq \y \leq \x\quad \mbox{and}\quad A\y\geq \zero.
\]
If $\y$ is maximal, i. e. there is no unit vector $\u\leq \x-\y$ such that $\y':=\y+\u$ satisfies the above inequalities, then  $\x':=\x-\y$ is another solution of \srsap.
\label{lem:other_sol}
\end{lemma}
\begin{proof} We have to show that $\zero\leq A\x' + \b\leq \c$. The inequality on the right is trivial:
$$A\x' + \b \leq (A\x' + \b) +A\y=A\x+\b\leq \c.$$
If the remaining inequality $A\x' + \b\geq \zero$ does not hold then there is some $i\in \left\{ 1, \dots, n \right\}$ such that $(A\x'+\b)_i\leq -1$ which implies
$$(A\x')_i \leq -1-(\b)_i\leq -1.$$
Hence there exists a unit vector $\u\leq\x'$ such that $(A\u)_i=-1$.
We claim that $A(\y+\u)\geq \zero$, which contradicts the maximality of $\y$.
As regards the $i$th entry,
$$(A(\y+\u))_i=(A\x+\b)_i-(A\x'+\b)_i+(A\u)_i\geq 0+1-1=0.$$
For $l\not=i$, due to the fundamental assumption on the matrix $A$, we have $(A\u)_l\geq 0$ and the proof is complete.
\end{proof}

In order to prove Theorem \ref{thm:SRT_iff_OP},
it suffices to show that if \srsap is solvable then \pegar is solvable as well. The proof is by induction on the the number of moves $|\x|$ of  the solution of \srsap $\x$. If $|\x|=1$ then we set $\u_1:=\x$ and we are done.
Assume that $|\x|>1$ and any $\x'$ solution of \srsap with $|\x'|<|\x|$ can be transformed into a solution of \pegar.

Let $\x_0=\zero$. We need to show that if we are at the step $k$ with $0\leq k<|\x|$ and
$\zero\leq A\x_k+\b\leq \uno$ then we are able to move forward to step $k+1$:
there exists a unit vector $\u\leq \x-\x_k$ such that
$\zero\leq A\x_{k+1}+\b\leq \uno$
where $\x_{k+1}:=\x_k+\u$.

We divide the proof of this fact into two claims.

\begin{claim}
There is a unit vector $\u$ such that $\u\leq \x-\x_k$ and  $A(\x_{k}+\u)+\b\geq \zero$.
\label{claim:one}
\end{claim}
\begin{proof}
    Let $\z:=\x-\x_k$, then $\z \geq 0$ and $\z\neq\zero$.

\noindent If $A\z\geq \zero$ then
    we extend $\z$ to a maximal vector $\y$ such that
    $\z\leq \y \leq \x$ and $A\y\geq \zero$.
    By Lemma \ref{lem:other_sol}, $\x':=\x-\y$ is another solution
    for \srsap with $|\x'|<|\x|$ and, by the induction hypothesis, it can be transformed into a solution of \pegar.

\noindent On the other hand, if $A\z\geq \zero$ does not hold,
    then there is some $i\in \left\{ 1,\dots, n \right\}$ such that
    $(A\z)_i\leq -1$ and
    there is a unit vector $\u\leq\z$ such that $(A\u)_i=-1$.
    If $l\not=i$ it follows that
    $$(A(\x_{k}+\u)+\b)_l=(A\x_{k}+\b)_l+(A\u)_l\geq (A\x_{k}+\b)_l\geq 0.$$
    As regards the $i$th entry, $A\x+\b\geq \zero$ implies
    $$(A\x_k+\b)_i=(A\x+\b)_i-(A\z)_i\geq 0+1=1.$$
    and therefore
    $$(A(\x_{k}+\u)+\b)_i=(A\x_{k}+\b)_i+(A\u)_i=
    (A\x_{k}+\b)_i-1\geq 1-1=0.$$
    Hence $A(\x_{k}+\u)+\b\geq \zero$.
\end{proof}

\begin{claim}
There is a unit vector $\u$ such that $\u\leq \x-\x_k$ and $\zero\leq A(\x_{k}+\u)+\b\leq \uno$.
    \label{claim:two}
\end{claim}
\begin{proof}
    Let $\U$ be the set of unit vectors $\u\leq\x-\x_k$ such that
    $A(\x_{k}+\u)+\b\geq \zero$. The set $\U$ is not empty by Claim \ref{claim:one}. If Claim \ref{claim:two} holds then the proof is complete. Otherwise
    given any $\u\in \U$ there is some $i\in \left\{
    1,\dots, n \right\}$ such that
    \[
        (A\x_{k}+\b)_i=1 \quad \mbox{and}\quad (A\u)_i=1.
    \]
    On the other hand, $\zero\leq A\x+\b\leq \c$ implies
    $$(A(\x-\x_k-\u))_i=(A\x+\b)_i-(A\x_k+\b)_i-(A\u)_i\leq 1-1-1=-1,$$
    and there is a unit vector $\v\leq \x-\x_k-\u$ such that $(A\v)_i=-1$. Hence
    \[
        (A(\x_k+\v)+\b)_i=(A\x_k+\b)_i+(A\v)_i=1-1=0
    \]
    which means that $\v \in \U$.

    By iterating this process, we generate a sequence of elements of $\U$.
    Since $\U$ is finite, this sequence has a minimal cycle,
    $\u_1,\dots,\u_{r-1},\u_r=\u_1$ and,
    by construction, for $2\leq t\leq r$, $(A\u_{t})_i=-1$ implies that $(A\u_{t-1})_i=1$.
    Setting $\z:=\sum_{t=2}^{r}\u_t$, we have
    $\zero \leq \z\leq \x$, $A\z\geq \zero$ and $\z \neq
    \zero$.
    Now, as in Claim \ref{claim:one}, we extend $\z$ to a maximal vector $\y$.
    Then, by Lemma \ref{lem:other_sol}, $\x':=\x-\y$ is another solution for
    \srsap with $|\x'|<|\x|$ and finally we apply the induction hypothesis.
\end{proof}

Finally, observe that the proof is constructive and can be easily turned into a
polynomial-time algorithm that performs the transformation.

\subsection{Explicit Solution to the Solitaire-Army Problem of Figure \ref{fig:peg}}
\label{apx:moves}

Given a configuration of pegs on a board, we denote a jump with
($x$,$y$,$d$), where $d \in \left\{ \uparrow, \downarrow,
\leftarrow, \rightarrow \right\}$ indicates the direction of the jump that has to
be made by the peg in position $(x,y)$ of the board.
The moves that allows the peg army in Figure \ref{fig:peg} to reach the center are listed below. \smallskip

\noindent \textbf{Platoon A}:
(-7,-2,$\rightarrow$), (-9,-2,$\rightarrow$), (-11,-2,$\rightarrow$), (-13,-2,$\rightarrow$), (-10,-4,$\uparrow$), (-10,-2,$\rightarrow$), (-10,0,$\downarrow$), (-12,0,$\rightarrow$), (-11,2,$\downarrow$), (-12,-1,$\rightarrow$), (-11,-2,$\rightarrow$), (-8,-2,$\rightarrow$), (-8,-4,$\uparrow$), (-9,-2,$\rightarrow$), (-6,-2,$\rightarrow$), (-6,0,$\downarrow$), (-6,2,$\downarrow$), (-6,4,$\downarrow$), (-8,1,$\rightarrow$), (-10,1,$\rightarrow$), (-7,3,$\downarrow$), (-7,5,$\downarrow$), (-9,2,$\rightarrow$), (-8,4,$\downarrow$), (-8,-1,$\rightarrow$), (-10,-1,$\rightarrow$), (-7,1,$\downarrow$), (-7,3,$\downarrow$), (-9,0,$\rightarrow$), (-11,0,$\rightarrow$), (-8,2,$\downarrow$), (-7,-2,$\rightarrow$), (-5,-2,$\rightarrow$), (-7,-1,$\rightarrow$), (-7,1,$\downarrow$), (-9,0,$\rightarrow$), (-8,-1,$\rightarrow$), (-6,-1,$\rightarrow$), (-7,0,$\rightarrow$), (-6,2,$\downarrow$), (-6,0,$\rightarrow$), (-4,0,$\downarrow$), (-4,-2,$\rightarrow$). \smallskip

\noindent \textbf{Platoon B}: symmetric moves of \textbf{Platoon A}. \smallskip

\noindent \textbf{Platoon C}:(-8,-6,$\uparrow$), (-10,-5,$\rightarrow$), (-12,-3,$\rightarrow$), (-10,-3,$\rightarrow$), (-9,-7,$\uparrow$), (-9,-5,$\uparrow$), (-7,-3,$\rightarrow$), (-6,-5,$\uparrow$), (-7,-5,$\uparrow$), (-6,-3,$\rightarrow$), (-8,-3,$\rightarrow$), (-8,-5,$\uparrow$), (-9,-3,$\rightarrow$), (-7,-3,$\rightarrow$), (-5,-3,$\rightarrow$). \smallskip

\noindent \textbf{Platoon D}: symmetric moves of \textbf{Platoon C}. \smallskip

\noindent \textbf{Platoon E}: (-2,-7,$\uparrow$), (-3,-7,$\uparrow$), (-4,-7,$\uparrow$), (-6,-6,$\rightarrow$), (-5,-8,$\uparrow$), (-5,-6,$\rightarrow$), (-6,-8,$\uparrow$), (-7,-6,$\rightarrow$), (-7,-8,$\uparrow$), (-8,-8,$\uparrow$), (-8,-6,$\rightarrow$), (-6,-6,$\rightarrow$), (-3,-6,$\uparrow$), (-4,-6,$\uparrow$), (-4,-4,$\rightarrow$), (-2,-5,$\uparrow$). \smallskip

\noindent \textbf{Platoon F}: symmetric moves of \textbf{Platoon E}. \smallskip

\noindent \textbf{Platoon G}: (-3,-9,$\uparrow$), (-5,-9,$\rightarrow$), (-7,-9,$\rightarrow$), (-6,-11,$\uparrow$), (-6,-9,$\rightarrow$), (-3,-10,$\uparrow$), (-5,-10,$\rightarrow$), (-5,-12,$\uparrow$), (-4,-12,$\uparrow$), (-3,-11,$\uparrow$), (-5,-10,$\rightarrow$), (-4,-9,$\uparrow$), (-4,-7,$\rightarrow$), (-3,-9,$\uparrow$), (1,-7,$\uparrow$), (1,-9,$\uparrow$), (3,-9,$\leftarrow$), (5,-9,$\leftarrow$), (7,-9,$\leftarrow$), (6,-11,$\uparrow$), (6,-9,$\leftarrow$), (4,-9,$\leftarrow$), (3,-11,$\uparrow$), (5,-10,$\leftarrow$), (5,-12,$\uparrow$), (4,-12,$\uparrow$), (4,-14,$\uparrow$), (3,-8,$\leftarrow$), (3,-10,$\uparrow$), (5,-10,$\leftarrow$), (4,-8,$\leftarrow$), (1,-8,$\uparrow$), (1,-10,$\uparrow$), (3,-10,$\leftarrow$), (2,-12,$\uparrow$), (4,-12,$\leftarrow$), (3,-14,$\uparrow$), (1,-11,$\uparrow$), (1,-13,$\uparrow$), (1,-15,$\uparrow$), (3,-15,$\leftarrow$), (1,-16,$\uparrow$), (3,-12,$\leftarrow$), (1,-12,$\uparrow$), (0,-7,$\uparrow$), (0,-9,$\uparrow$), (-2,-9,$\rightarrow$), (-1,-11,$\uparrow$), (-1,-13,$\uparrow$), (-1,-15,$\uparrow$), (-3,-12,$\rightarrow$), (-2,-14,$\uparrow$), (-4,-13,$\rightarrow$), (-3,-10,$\rightarrow$), (2,-8,$\leftarrow$), (2,-10,$\uparrow$), (1,-10,$\uparrow$), (0,-8,$\uparrow$), (2,-8,$\leftarrow$), (0,-9,$\uparrow$), (0,-11,$\uparrow$), (-2,-11,$\rightarrow$), (0,-12,$\uparrow$), (-2,-12,$\rightarrow$), (0,-13,$\uparrow$), (2,-13,$\leftarrow$), (0,-14,$\uparrow$), (2,-14,$\leftarrow$), (-2,-13,$\rightarrow$), (-2,-8,$\rightarrow$), (-1,-10,$\uparrow$), (-1,-7,$\uparrow$), (-3,-7,$\rightarrow$), (-1,-8,$\uparrow$), (1,-6,$\uparrow$), (0,-6,$\uparrow$), (0,-8,$\uparrow$), (0,-10,$\uparrow$), (0,-12,$\uparrow$), (0,-14,$\uparrow$), (0,-16,$\uparrow$), (-2,-16,$\rightarrow$), (0,-17,$\uparrow$), (0,-15,$\uparrow$), (0,-13,$\uparrow$), (0,-11,$\uparrow$), (0,-9,$\uparrow$), (0,-7,$\uparrow$), (-1,-6,$\uparrow$), (0,-5,$\uparrow$). \smallskip

\noindent \textbf{Finale}: (-3,-3,$\rightarrow$), (3,-3,$\leftarrow$), (-1,-4,$\uparrow$), (1,-4,$\uparrow$), (-2,-2,$\rightarrow$), (0,-3,$\uparrow$), (2,-2,$\leftarrow$), (0,-2,$\uparrow$).

\bibliography{peg-solitaire}
	
\end{document}